\RequirePackage{tikz-cd}

\documentclass[sn-basic]{sn-jnl}

\usepackage{graphicx}
\usepackage{calligra,mathrsfs,mathtools}
\usepackage{bm,amsmath,amssymb}
\usepackage{enumitem}
\usepackage{amsthm}
\usepackage{extarrows}
\usepackage{faktor}
\usepackage[normalem]{ulem}
\usepackage{tikz-cd}
\newcommand{\com}[1]{\textcolor{red}{#1}}

\jyear{2022}%

\theoremstyle{thmstyleone}%
\theoremstyle{thmstyletwo}%

\theoremstyle{thmstylethree}%

\theoremstyle{definition}
\newtheorem{thm}{Theorem}[section]
\numberwithin{thm}{section}
\newtheorem{defn}[thm]{Definition}
\newtheorem{ex}[thm]{Example}
\newtheorem{rmk}[thm]{Remark}
\newtheorem{prop}[thm]{Proposition}
\newtheorem{corr}[thm]{Corollary}
\newtheorem{lem}[thm]{Lemma}
\numberwithin{equation}{section}

\begin{document}

\title{Semigroup models for biochemical reaction networks}


\author{\fnm{Dimitri} \sur{Loutchko}} \email{d.loutchko@edu.k.u-tokyo.ac.jp}

\affil{\orgdiv{Institute of Industrial Science}, \orgname{The University of Tokyo}, \orgaddress{\street{4-6-1, Komaba}, \city{Meguro-ku}, \postcode{153-8505}, \state{Tokyo}, \country{Japan}}}

\abstract{
The catalytic reaction system (CRS) formalism by Hordijk and Steel is a versatile method to model autocatalytic biochemical reaction networks.
It is particularly suited, and has been widely used, to study self-sustainment and self-generation properties.
Its distinguishing feature is the explicit assignment of a catalytic function to chemicals that are part of the system.
In this work, it is shown that the subsequent and simultaneous catalytic functions gives rise to an algebraic structure of a semigroup with the additional compatible operation of idempotent addition and a partial order.
The aim of this article is to demonstrate that such semigroup models are a natural setup to describe and analyze self-sustaining CRS.
The basic algebraic properties of the models are established and the notion of the function of any set of chemicals on the whole CRS is made precise.
This leads to a natural discrete dynamics on the network, which results from iteratively considering the self-action on a set of chemicals by its own function.
The fixed points of this dynamics are proven to correspond to self-sustaining sets of chemicals, which are functionally closed.
Finally, as the main application, a theorem on the maximal self-sustaining set and a structure theorem on the lattice of functionally closed self-sustaining sets of chemicals are proven.
}




\maketitle

\section{Introduction}

In order to describe and understand the organization of very large biochemical reaction networks, in particular with respect to their evolutionary origins, the classical models developed by \cite{Horn1972} and \cite{Feinberg1987,feinberg2019} based on systems of ordinary differential equations are often not the optimally suited mathematical framework.
Although they are most capable of describing a wide plethora of biochemical phenomena, as is vividly presented in the books by \cite{beard2008}, \cite{mikhailov2017} and many others, these models are limited by their relatively high computational cost and their reliance on the precise knowledge of kinetic parameters.
Several other approaches have been proposed with the primary aim of capturing general and universal properties of biochemical reaction networks, with a particular focus on the functional organization and self-sustainment as well as self-generation properties of the reaction networks of living and evolving systems.
Examples include $(M,R)$-systems by \cite{Rosen1958}, hypercycles introduced by \cite{Eigen1971}, autopoetic systems studied by \cite{Varela1974}, chemotons by \cite{Ganti1975}, autocatalytic sets by \cite{Kauffman1986} and catalytic reaction systems developed by \cite{Hordijk2004}.
The common property of such models is that they focus on the catalytic function of network reactions by chemicals which are themselves produced by reactions within the network.
They usually do not require kinetic details, but only the knowledge of the occurring reactions, together with data on catalysis.
An in-depth discussion and comparison of such approaches was given by \cite{Hordijk2018}.\\

The construction of the semigroup models presented in this article is based on the formalism of catalytic reaction systems (CRS) introduced by \cite{Hordijk2004,Hordijk2018,Hordijk2011}, which is a generalization of autocatalytic sets by \cite{Kauffman1986} and is broad enough to encompass several of the aforementioned approaches.
A CRS is given by the datum of a chemical reaction network, i.e. by a finite set of chemicals and specified chemical reactions, together with the assignment of catalytic functionality to certain chemicals.
The main purpose of this article is to flesh out the algebraic structure which is inherent in the CRS formalism given by the occurrence of simultaneous and subsequent reactions and to investigate the properties of the resulting semigroup models.
In particular, the construction allows to make the notion of the catalytic function of chemicals and of sets of chemicals on the whole CRS precise.
In order to illustrate that the semigroup models provide a well-suited mathematical language for CRS, it is shown how self-sustaining subsets of chemicals can be characterized in a concise manner and how the largest self-sustaining subset of chemicals can be determined for any CRS.
Moreover, within the semigroup formalism it is natural to consider self-sustaining sets of chemicals which cannot produce chemicals not already contained in the respective set.
Such sets have not been considered in the literature thus far and the term of {\it functionally closed set} of chemicals is coined here for them.
The lattice of all functionally closed sets of the CRS is characterized and its potential applications to the analysis of CRS which model real biological systems are discussed.
In a nutshell, the semigroup models provide natural tools to concisely state and solve problems in CRS theory and to find new meaningful constructions within the theory.\\

Besides their applications to CRS theory, semigroup models open up the field of biochemical reaction networks to be tackled with algebraic methods.
For example, the equivalence of finite semigroups and finite automata suggests to investigate the computational capabilities of chemical reaction networks from this abstract point of view.
Moreover, it will be interesting and rewarding to construct a purely algebraic description of semigroup models without the need to consider their representation as maps on a finite set.
One step in this direction, which is proven in this work, is the fact that the semigroup models of CRS with a nontrivial self-sustaining sets of chemicals cannot be nilpotent.
This excludes a large class of potential semigroup models to be considered in inverse problems in the future.\\

\paragraph{Mathematical outline} 

In Section \ref{sec:CRS}, the CRS formalism is introduced following \cite{Hordijk2004,Hordijk2017}.
More precisely, a CRS uses the datum of a chemical reaction network, i.e. a finite set of chemicals $X$ together with a finite set of net chemical reactions $R$.
In addition, a set of catalysis data $C \subset X \times R$ is specified by stipulating that for each $(x,r) \in C$, the reaction $r$ is catalyzed by the chemical $x$.
Finally, a subset $F \subset X$, called food set, of constantly supplied chemicals is given.
The kinetic rate constants of the chemical reaction network are not part of the datum of a CRS.
Within the CRS formalism, the notion of self-sustainment has been formalized under the name of pseudo-RAF by \cite{Steel2015} and \cite{Hordijk2017}.\\

In Section \ref{sec:SemigroupBasics}, the catalytic function of chemicals is given the structure of a semigroup, which is additionally equipped with a partial order and an idempotent addition.
To begin with, to each reaction $r \in R$, a function $\phi_r$ is assigned as the set-map $\phi_r: \mathfrak{X} \rightarrow \mathfrak{X}$ on the power set $\mathfrak{X} :=\mathcal{P}(X_F)$ of non-food chemicals $X_F = X \setminus F$.
The function $\phi_r$ gives the set of non-food products of the reaction $r$ if and only if the set of non-food reactants is contained in its argument.
Such functions have an idempotent addition via $(\phi_r + \phi_{r'})(Y) = \phi_r(Y) \cup \phi_{r'}(Y)$ for any $Y \subset X_F$ and any $r,r' \in R$.
The function $\phi_x$ of a chemical $x \in X$ is defined as the sum of the reactions catalyzed by it:
\begin{equation*}
 \phi_x = \sum_{(x,r) \in C} \phi_r.
\end{equation*}
The semigroup model
\begin{align*}
 \mathcal{S} = \langle \phi_x \rangle _{x\in X}
\end{align*}
of a CRS is generated by the functions $\{\phi_x\}_{x\in X}$ through addition $+$ and through the composition of functions $\circ$.
Thus, $\mathcal{S}$ is a semigroup with respect to both $+$ and $\circ$ and hence it is called a {\it semigroup model} of the CRS.
For any subset $Y \subset X_F$, the semigroup model $\mathcal{S}(Y)$ is defined as $\mathcal{S}(Y) = \langle \phi_x \rangle _{x\in Y \cup F}$.
Moreover, $\mathcal{S}$ is endowed with a partial order given by $\phi \leq \psi$ iff $\phi(Y) \subset \psi(Y)$ for all $Y \subset X_F$.
The unique maximal element of $\mathcal{S}(Y)$ is denoted by $\Phi_Y$ - it represents all catalytic functionality that can be exhibited by the set $Y$ on the whole CRS.
In Section \ref{sec:SElementaryProperties}, elementary properties of the semigroup models are established.
In a nutshell, these properties ensure that the two operations and the partial order on $\mathcal{S}$ are compatible.\\

In Section \ref{sec:RAF}, the semigroup models are used to derive theorems characterizing self-sustaining CRS and subCRS and their corresponding sets of chemicals.
Theorem \ref{thm:PRAF} concisely states that a CRS is self-sustaining if and only if the set of chemicals $X_F$ is able to fully reproduce itself.
In the semigroup formalism this reads as
\begin{equation*}
 \Phi_{X_F}(X_F) = X_F.
\end{equation*}
It is then shown that each self-sustaining set of chemicals $X'_F \subset X_F$ satisfies $X'_F \subset \Phi_{X'_F}(X'_F)$ (Corollary \ref{corr:subPRAF}) and that the equality $X'_F = \Phi_{X'_F}(X'_F)$ is a sufficient condition to be a self-sustaining set of chemicals (Proposition \ref{prop:PRAFsuff}).
Finally, a discrete dynamics on the state space $\mathfrak{X}$ is defined by the self-action of the chemicals on themselves, i.e. by $Y \mapsto \Phi_Y(Y)$ for $Y \in \mathfrak{X}$.
It is shown that the dynamics with the initial condition $X_F$ leads to a fixed point, denoted by $X_F^{*s}$ and, using a combination of the tools developed in this article, it is proven that the maximal self-sustaining set chemicals of any CRS is given by this fixed point (Theorem \ref{thm:maxPRAF}).\\

In Section \ref{sec:functionalClosure}, new concepts are introduced into CRS theory, which are motivated by the tools from the preceding sections.
The containment of a self-sustaining set of chemicals $X_F'$ in the set of chemicals which it generates, i.e. $X_F' \subset \Phi_{X'_F}(X'_F)$, suggests that such self-sustaining sets should not be observable in real biological systems.
Instead, they will continue to produce chemicals according to the dynamics ${X_F' =: Y_0 \mapsto \Phi_{Y_0}(Y_0) =: Y_1 \mapsto \Phi_{Y_1}(Y_1) \mapsto \dotsm}$, which is shown to stabilize at the self-sustaining set of chemicals $X_F'^{*s}$ (Proposition \ref{prop:increasingFixedPoint}).
The set $X_F'^{*s}$ is called the functional closure of $X_F'$ and it is conjectured that within a CRS such sets comprise biologically relevant functional modules of a CRS.
To determine the lattice of all functionally closed sets of chemicals, a restricted dynamics on $\mathfrak{X}$ is introduced via $Y \mapsto Y \cap \Phi_Y(Y)$.
This dynamics has a fixed point $Y_0^{*rs}$ for any initial condition $Y_0 \in \mathfrak{X}$.
The desired lattice is given by the set of fixed points
\begin{equation*}
    \mathfrak{M} := \bigcup_{i=0}^{\mid X_F^{*s} \mid} \mathfrak{M}^{i} \subset \mathfrak{X},
\end{equation*}
where the $\mathfrak{M}^i$ are defined by the recursion 
\begin{align*}
    \mathfrak{M}^0 &:= \{ X_F^{*s} \} \\
    \mathfrak{M}^{i+1} &:= \bigcup_{Y \in  \mathfrak{M}^{i}(X_F^{*s})} \mathfrak{M}(Y)
\end{align*}
for $i \geq 0$ and where
\begin{equation*}
   \mathfrak{M}(Y) := \left\{  (Y\setminus\{y\})^{*rs} \right\}_{y \in Y} \subset \mathfrak{X}
\end{equation*}
for any $Y \in \mathfrak{X}$.
This is proven in Theorem \ref{thm:M} and provides an explicit description, which can be algorithmically implemented, of the lattice.\\

In a companion article by \cite{loutchko2022arxiv}, the self-generating properties of CRS are investigated in an analogous fashion.
Although the logical structure of the line of reasoning is analogous to the one presented here, the proofs themselves require more technical setup, which heavily relies on a representation of the semigroup elements as decorated rooted trees.
In this regard, Section \ref{sec:RAF} of this article can be seen as a blueprint to the work in \cite{loutchko2022arxiv}.
The current article is aimed as an introduction to semigroup models of CRS, whereas the companion article is focused on their structure and provides more technical tools.

\section{The CRS Formalism} \label{sec:CRS}

\subsection{Basic Notions}

The catalytic reaction system (CRS) formalism is introduced following \cite{Hordijk2004}.
A chemical reaction network is a finite set of chemicals $X$ together with a finite set of net chemical reactions $R$.
Thereby, a reaction $r \in R$ is given by a pair of sets $r=(\textrm{dom}(r),\textrm{ran}(r))$ which satisfy $\textrm{dom}(r),\textrm{ran}(r) \subset X$ and $\textrm{dom}(r) \cap \textrm{ran}(r) = \emptyset$ and are called the {\it domain} and the {\it range} of the reaction\footnote{Usually, a net chemical reaction is denoted as $r: a_1 A_1 + a_2 A_2 + ... + a_n A_n \longrightarrow b_1 B_1 + b_2 B_2 + ... + b_m B_m$, where $a_i,b_j \in \mathbb{N}$ and $A_i,B_j \in X$ such that $A_i \neq B_j$ for all $i=1,...,n$ and $j=1,...,m$.
In the definition used here, the stoichiometric information is discarded and $\textrm{dom}(r) = \{A_i\}_{i=1}^n$ and $\textrm{ran}(r) = \{B_j\}_{j=1}^m$ are used instead.}
The latter condition ensures that one is dealing with a {\it net} chemical reaction. 
Catalysis and autocatalysis is treated in Definition \ref{def:CRS} below by giving additional data on the catalyzed reactions.
The support $\text{supp}(r)$ of a reaction is defined as
\begin{equation*}
 \text{supp}(r) = \text{dom}(r) \cup \text{ran}(r).
\end{equation*}
In addition, for a set of reactions $R' \subset R$, the domain, range and support are defined as
\begin{align*}
\text{dom}(R') &= \bigcup_{r \in R'} \text{dom}(r), \\
\text{ran}(R') &= \bigcup_{r \in R'} \text{ran}(r), \\
\text{supp}(R') &= \bigcup_{r \in R'} \text{supp}(r) = \text{dom}(R') \cup \text{ran}(R').
\end{align*}
Classically, a set of rate constants together with the stoichiometry of each $r \in R$ would lead to a system of coupled ordinary differential equations for the time evolution for the concentration of each chemical $x \in X$ given by a kinetic model such as mass action kinetics.
However, the CRS formalism does not utilize this detailed kinetic information but instead emphasizes the catalytic function of the chemicals in $X$.
Moreover, a food set $F \subset X$ is specified.
This set is to be thought of as the set of externally supplied chemicals, which are always readily available to the system.
The abbreviation
\begin{equation*}
    X_F := X \setminus F
\end{equation*}
for the set of all non-food chemicals will be used.
Often, a subset of $X_F$ will be denoted by $X'_F$, which implicitly carries the information of the corresponding subset $X' = X'_F \cup F$ of $X$.

\begin{defn} \label{def:CRS}
A {\it catalytic reaction system} (CRS) is a tuple $(X,R,C,F)$ where $X$ is a finite set of chemicals, $R$ is a finite set of reactions, $C \subset X \times R$ is the catalysis data, and $F \subset X$ is the constantly present food set.
For a pair $(x,r) \in C$, the reaction $r$ is said to be catalyzed by $x$.
The food set is required to satisfy the following closure property:
\begin{enumerate} [label=(C),leftmargin=1cm]
    \item All reactions $r \in R$ with a catalyst in $F$ must involve reactants outside of $F$, i.e. they must satisfy $\text{dom}(r) \cap X_F \neq \emptyset$.  \label{cond:C}
\end{enumerate}
If $X=F$, the CRS is said to be {\it trivial}.
\end{defn}

\noindent It is convenient to introduce the projection
\begin{align} \label{eq:projC}
    \pi_R: C \rightarrow R,
\end{align}
which gives all catalyzed reactions of the CRS.

\begin{rmk}
In the CRS literature it is common to consider the triple $(X,R,C)$ as the definition of a CRS and to give the food set $F \subset X$ as an additional datum, cf. \cite{Hordijk2004,Hordijk2011}.
The tuple $(X,R,C,F)$ from Definition \ref{def:CRS} would then be called a {\it CRS with food set}.
Because this article is concerned with tuples $(X,R,C,F)$ exclusively, they are simply called CRS for the sake of readability.
The closure condition \ref{cond:C} is not commonly used in the CRS literature.
However, it poses no serious restrictions but merely serves to exclude trivial cases of self-sustainment, wherein chemicals in $X_F$ are produced purely from the food set by reactions with catalysts in the food set.
\end{rmk}

Following \cite{Bonchev2012}, a CRS can be represented as a directed bipartite graph with a partition of the edges into edges corresponding to reactions and edges corresponding to catalysis.
As an example, consider the graph in Fig. \ref{fig:example1}.
The vertices represented by solid disks correspond to the chemicals $X$ and the vertices represented by circles correspond to reactions $R$.
Solid directed edges from chemicals to a reaction indicate that the respective chemicals are the reactant of the reaction, whereas solid directed edges from a reaction to chemicals are present for the products of the respective reaction.
Dashed edges from a chemical $x$ to a reaction $r$ represent elements $(x,r)$ in the catalysis data $C$.
The food set is indicated by a circle, which encloses the chemicals belonging it.
The stoichiometry of the reactions is not shown in the graph.

\begin{figure}[htb]
  \centering
  \includegraphics[scale=0.25]{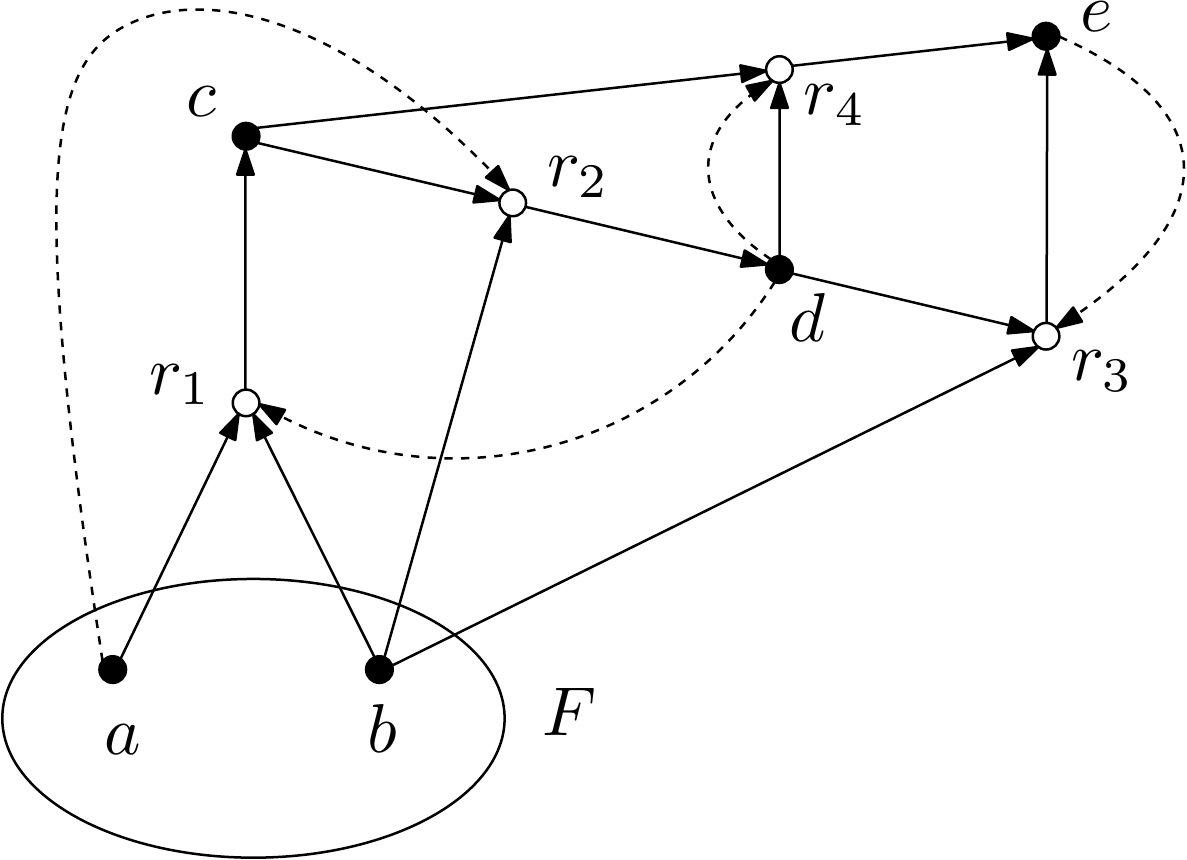}
  \caption[Example of a catalytic reaction system]{Example of a graphical representation of a CRS. 
  The CRS consists of five chemicals $X=\{a,b,c,d,e\}$ and four reactions $a+b\rightarrow c$, $c+b \rightarrow d$, $b+d \rightarrow e$ and $c+d \rightarrow e$, which are catalyzed by $d$, $a$, $e$ and $d$, respectively.
  The food set is given by $F = \{a,b\}$.}
  \label{fig:example1}
\end{figure}

\begin{defn} \label{def:subCRS}
 A tuple $(X',R',C',F)$ is said to be a {\it subCRS} of the full CRS $(X,R,C,F)$ if the following properties are satisfied

\begin{enumerate} [label=(C\arabic*),leftmargin=1cm]
    \item $X' \subset X$,  \label{cond:S1}
    \item The food set $F \subset X'$ is identical to the one of the full CRS. \label{cond:S0}
    \item $R' \subset R\mid_{X'}$, \label{cond:S2}
    \item $C' = C \cap (X' \times R') \subset C\mid_{X'}$, \label{cond:S3}
\end{enumerate}
where the restrictions $R\mid_{X'}$ and $C\mid_{X'}$ are given by
\begin{align*}
    R\mid_{X'} &= \{r \in R \text{ with supp}(r) \subset X'\}, \\
    C\mid_{X'} &=  C \cap (X' \times R\mid_{X'}). 
\end{align*}
\end{defn}

The sets $R\mid_{X'}$ and $C\mid_{X'}$ represent the maximal set of reactions supported on $X'$ and the maximal catalysis data with reactions supported on $X'$ and catalysts in $X'$, respectively.

Note that a subCRS is determined by solely the two subsets $X' \subset X$ and $R' \subset R\mid_{X'}$. 
The set $C'$ is the maximal catalysis data for reactions in $R'$ that are catalyzed by some chemical in $X'$.
A subCSR $(X',R',C',F)$ always satisfies condition \ref{cond:C}:
By Definition \ref{def:subCRS}, a reaction $r \in R'$ satisfies $\textrm{dom}(r) \subset X'$ and therefore $\text{dom}(r) \cap X'_F = \text{dom}(r) \cap X_F$ holds.
If it has a catalyst in $F$, then $\text{dom}(r) \cap X_F \neq \emptyset$ holds by condition \ref{cond:C} for the full CRS $(X,R,C,F)$, which implies $\text{dom}(r) \cap X'_F \neq \emptyset$.
Therefore, a subCRS is always a CRS.\\

The subCRS of a CRS are ordered by inclusion, i.e. two subCRS are related by $(X',R',C',F) \leq (X'',R'',C'',F)$ iff $X' \subset X''$ and $R' \subset R''$ are satisfied.
For a set of non-food chemicals $X'_F \subset X_F$, the maximal subCRS with the set of chemicals $X'$ is given by $(X',R\mid_{X'},C\mid_{X'},F)$.
It is called the {\it subCRS generated by $X'_F$} and denoted by
\begin{equation} \label{eq:CRSX}
    CRS(X'_F):= (X',R\mid_{X'},C\mid_{X'},F).
\end{equation}

\subsection{Catalytic and self-sustainment properties}

Self-sustaining CRS are characterized by the property that all non-food chemicals are products of catalyzed reactions which use only chemicals produced by the reactions themselves and the food set.
Self-sustaining CRS have been discussed in \cite{Hordijk2017} and \cite{Steel2015}, where they are called pseudo-RAF.
In this section, self-sustaining CRS and, based on them, self-sustaining sets of chemicals are defined.
Finally, the precise connection to the definition of pseudo-RAF as given by \cite{Hordijk2017} is discussed.

\begin{defn} \label{def:PRAF}
    A CRS $(X,R,C,F)$ is called {\it self-sustaining} if the following condition is satisfied:
    \begin{enumerate} [label=(S),leftmargin=1cm]
    \item There exists a set of reactions $R' \subset \pi_R(C)$ such that $X_F \subset \textrm{ran}(R')$ and $\textrm{dom}(R') \subset \textrm{ran}(R') \cup F$.
    \label{cond:PR}
\end{enumerate} 
\end{defn}

Note that it is not necessary that all reactions in $R$ are catalyzed but only that the CRS has enough catalyzed reactions to satisfy condition \ref{cond:PR}.
A subCRS is said to be self-sustaining if it satisfies the Definition \ref{def:PRAF}.
For a given set of chemicals $X' \subset X$, the level of catalysis can differ between subCRS which have $X'$ as the set of chemicals.
If all catalyzed reactions of the full CRS $(X,R,C,F)$ which have a catalyst and support in $X'$ are included in $R'$, then the respective subCRS is said to be closed.

\begin{defn} \label{def:CRAF}
A subCRS $(X',R',C',F)$ is {\it closed} if its set of reactions satisfies
\begin{equation*}
    \pi_{R'}(C') = \pi_{R\mid_{X'}}(C\mid_{X'}),
\end{equation*}
i.e. if all catalyzed reactions with catalyst and support in $X'$ are included in $R'$.
\end{defn}
\noindent The CRS $(X,R,C,F)$ is closed by definition.
Moreover, it has a maximal, possibly trivial, self-sustaining subCRS given by the union of all its self-sustaining subCRS (the union of two given subCRS is the smallest subCRS for which both are subCRS).
This maximal self-sustaining subCRS is necessarily closed.

A set of chemicals $X'_F \subset X_F$ generates $CRS(X'_F)$ via formula (\ref{eq:CRSX}) and thus the notion self-sustainment is inherited by it:
\begin{defn} \label{def:RAFsetofchemicals}
    A set of chemicals $X'_F \subset X_F$ is said to be a {\it self-sustaining set of chemicals} if the maximal subCRS generated by it, $CRS(X'_F)$, is self-sustaining.
    The empty set $\emptyset \subset X_F$ is said to be the trivial self-sustaining set of chemicals.
\end{defn}
\noindent If $X'_F$ is a self-sustaining set of chemicals, then $CRS_F(X'_F)$ is automatically closed due to the maximality of $R\mid_{X'}$ and $C\mid_{X'}$.
Therefore, the notion of self-sustainment on the level of chemicals does not encompass non-closed CRS.

\begin{ex}
The CRS in Fig. \ref{fig:example2} is self-sustaining in the sense that all non-food chemicals $X_F = \{a,b,c\}$ are produced by the set of reactions $R =\{r_1,r_2,r_3\}$.
The set of reactions satisfies $X_F \subset \textrm{ran}(R) = \{a,b,c\}$ and $\textrm{dom}(R) = \{a,b,c,d\} \subset \textrm{ran}(R) \cup F = \{a,b,c,d\}$.
Note that none of the chemicals can be generated from the food set alone.
\begin{figure}[htb]
  \centering
  \includegraphics[scale=0.25]{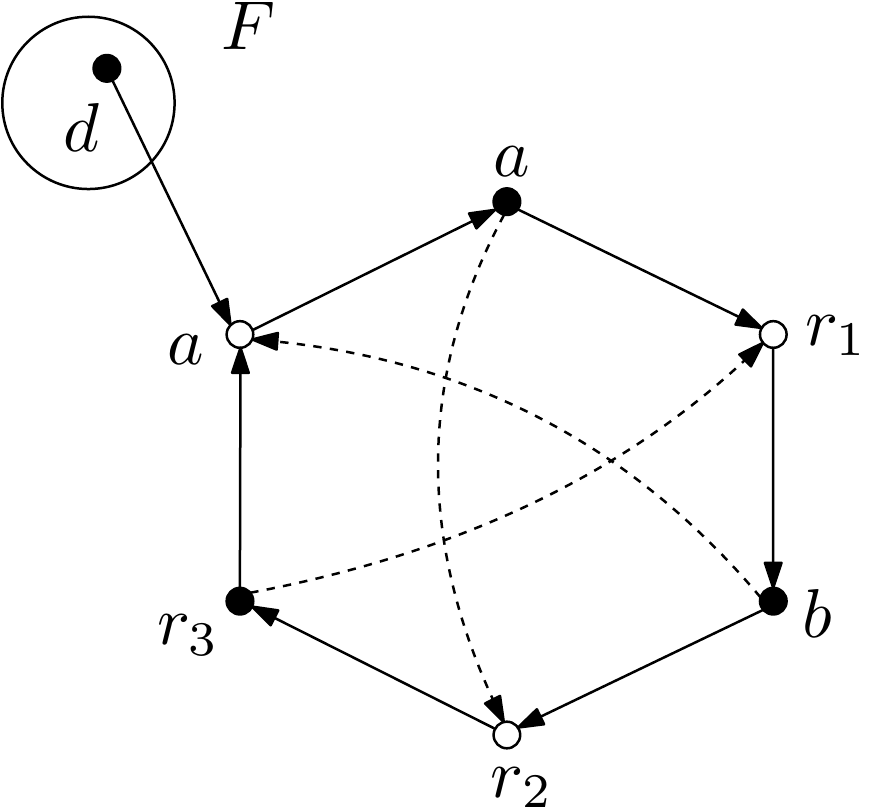}
  \caption{An example of a self-sustaining CRS.}
  \label{fig:example2}
\end{figure}
\end{ex}

\begin{rmk}[Relation to the notion of pseudo-RAF commonly used in the literature] \label{rmk:RAFliterature}
The definition of a pseudo-RAF given by \cite{Hordijk2017} (see also \cite{steel2013,Steel2015,Hordijk2015}) is based on a set $R' \subset \pi_R(C)$ of catalyzed reactions.
More precisely, a subset $R' \subset \pi_R(C)$ is called pseudo-RAF if $\textrm{dom}(R') \subset \textrm{ran}(R') \cup F$ holds.
Therefore, given a subCRS $(X',R',C',F)$, which is self-sustaining according to Definition \ref{def:PRAF}, the set $\pi_{R'}(C')$ is a pseudo-RAF set of reactions.
{\it Vice versa}, given a pseudo-RAF set of reactions $R' \subset R$, the subCRS $(X',R',C',F)$ with $X' := \textrm{supp}(R')$ and $C':= C \cap ( X' \times R')$ is self-sustaining according to Definition \ref{def:PRAF}.
Therefore, both definitions are equivalent modulo the inclusion of uncatalyzed reactions in the sets of reactions in the definition used in this article\footnote{This can be made precise by saying that two subCRS $(X',R'_1,C'_1,F)$ and $(X',R'_2,C'_2,F)$ with the same set of chemicals $X'$ are catalytically equivalent if $\pi_{R'_1}(C'_1) = \pi_{R'_2}(C'_2)$ holds.
The catalytical equivalence classes of subCRS are ordered by applying the partial order of subCRS to the minimal representatives of each catalytical equivalence class.}.

The self-sustaining sets of chemicals introduced in Definition \ref{def:RAFsetofchemicals} correspond to closed pseudo-RAF sets of reactions via $X'_F \mapsto R\mid_{X'}$ and, {\it vice versa}, a closed pseudo-RAF set of reactions $R'$ yields the corresponding self-sustaining set of chemicals as $X'_F := \textrm{supp}(R') \cap X_F$.
\end{rmk}

\section{The semigroup model of a CRS} \label{sec:SGM}

The catalytic function of the chemicals gives rise to an algebraic structure generated by the simultaneous and subsequent catalytic functionality.
In this section, this structure is formalized mathematically and leads to the notion of a semigroup model of a CRS.
Throughout this section, let the CRS $(X,R,C,F)$ be fixed.

\subsection{Construction} \label{sec:SemigroupBasics}

The state of the CRS is defined by the presence or absence of each of the non-food chemicals, i.e. by a subset $Y \subset X_F$ (the chemicals in the food set are always present).
Thus the state space $\mathfrak{X}$ of the CRS is defined to be the power set of $X_F$
\begin{align*}
    \mathfrak{X} := \mathcal{P}(X_F).
\end{align*}
The function of a given chemical $x \in X$ is defined via the reactions it catalyzes, i.e. by the way it acts on the state space $\mathfrak{X}$.
This definition is inspired by the work of \cite{Rhodes2010}.

\begin{defn} \label{def:func}
Let $(X,R,C,F)$ be a CRS with food set.
The {\it function} $\phi_r$ of a reaction $r \in R$ is the map
\begin{equation*}
 \phi_r : \mathfrak{X} \rightarrow \mathfrak{X}
\end{equation*}
given by
\begin{equation} \label{eq:phi_r}
  \phi_r(Y) =\begin{cases}
    \text{ran}(r) \cap X_F & \text{if $\text{dom}(r) \subset Y \cup F$}\\
	\emptyset & \text{else}
  \end{cases}
\end{equation}
for all $Y \subset X_F$ \footnote{Note that although $Y$ does not contain $F$, the constant presence of the food set is accounted for in the definition via formula (\ref{eq:phi_r}).}.

\noindent For any two maps $\phi, \psi: \mathfrak{X} \rightarrow \mathfrak{X}$, the sum $(\phi + \psi): \mathfrak{X} \rightarrow \mathfrak{X}$ and product $(\phi \circ \psi): \mathfrak{X} \rightarrow \mathfrak{X}$ are given by
\begin{align} \label{eq:addition}
  (\phi + \psi)(Y) &:= \phi(Y) \cup \psi(Y) \\
  \label{eq:multiplication}
  (\phi \circ \psi)(Y) &:= \phi (\psi(Y))
\end{align} for all $Y \subset X_F$.
The {\it function} $\phi_x: \mathfrak{X} \rightarrow \mathfrak{X}$ of a chemical $x \in X$ is defined as the sum over all reactions catalyzed by $x$, i.e.
\begin{equation} \label{eq:phi_x}
\phi_x = \sum_{(x,r) \in C} \phi_r.
\end{equation}

\end{defn}
The multiplication $\circ$ is the usual composition of maps and therefore associative.
The addition is associative, commutative and idempotent.
The two operations $\circ$ and $+$ have obvious interpretations in terms of the function of enzymes on a CRS:
The sum of two functions $\phi_x + \phi_y$ with $x,y \in X$ describes the {\it joint} or {\it simultaneous} function of two chemicals $x$ and $y$ on the CRS - it captures the reactions catalyzed by both $x$ and $y$.
The composition of two functions $\phi_x \circ \phi_y$ with $x,y \in X$ describes their {\it subsequent} function: first $y$ and then $x$ act on the state space by their respective catalytic functions.\\

Recall that the {\it full transformation semigroup} $\mathcal{T}(A)$ of a finite discrete set $A$ is the set of all maps $\{f:A \rightarrow A\}$, where the semigroup operation $\circ$ is the composition of maps.

\begin{defn} \label{def:semigroupModel}
Let $(X,R,C,F)$ be a CRS.
Its {\it semigroup model} $\mathcal{S}$ is the smallest subsemigroup of the full transformation semigroup $\mathcal{T}(\mathfrak{X})$ closed under $\circ$ and $+$ that contains $\{\phi_x\}_{x \in X}$ and the zero function given by $0(Y) = \emptyset$ for all $Y \subset X_F$.
The semigroup model is said to be {\it generated} by the set $\{\phi_x\}_{x \in X}$, indicated by the notation
\begin{equation*}
\mathcal{S} = \langle \phi_x \rangle_{x \in X}.
\end{equation*}
\end{defn}

\noindent As a subsemigroup of $\mathcal{T}(\mathfrak{X})$, the semigroup $\mathcal{S}$ is automatically a {\it finite} semigroup.\\
\begin{rmk} \label{rmk:algebra}
 The object $\mathcal{S}$ is called a semigroup model, because $\mathcal{S}$ is a semigroup with respect to both operations $\circ$ and $+$.
 The correct description of $\mathcal{S}$ in terms of universal algebra is an algebra of type $(2,2,0)$, cf. \cite{Almeida1995}.
 However, to avoid confusion with the more commonly used notions of an algebra (e.g. matrix algebras, operator algebras, algebras over rings, etc.), this terminology is not used.
 When referring to the actual semigroups $(\mathcal{S},\circ)$ and $(\mathcal{S},+)$, the term semigroup is used instead of semigroup model. 
\end{rmk}

\begin{rmk} \label{rmk:Rhodes}
 \cite{Rhodes2010} have introduced and studied similar semigroups derived from chemical reaction networks.
 These semigroups were constructed as subsemigroups of the full transformation semigroup on a set of metabolites induced by functions of enyzmes similar to the definition given here.
 However, in their work, the catalysts were not part of the network.
 More importantly, they did not define and use the operation of addition and also did not consider the natural partial order introduced in Remark \ref{rmk:PO}.
 Therefore, most results of this article cannot be obtained with the construction of Rhodes and Nehaniv.
 Moreover, in the definition given here, the chemicals which are not produced by a function $\phi \in \mathcal{S}$ disappear after the application of the function to any $Y \in \mathfrak{X}$, whereas they are retained in the model of Rhodes and Nehaniv.
 Yet, this property is crucial for obtaining the results in Section \ref{sec:RAF}.
\end{rmk}

\begin{rmk} \label{rmk:generators}
Generally, a map $\phi : \mathfrak{X} \rightarrow \mathfrak{X}$ is to be defined on all subsets $Y \subset X_F$, i.e. the assignment $Y \mapsto \phi(Y)$ needs to be given independently for all $Y \subset X_F$.
However, in the case of the constructed semigroup models, all maps $\phi \in \mathcal{S}$ respect the partial order on $\mathfrak{X}$ given by inclusion of sets, i.e.
\begin{equation} \label{eq:generators}
 Z \subset Y \subset X_F \implies \phi(Z) \subset \phi(Y)
\end{equation}
\noindent for all $Y,Z \subset X_F$.
This follows directly from the Definition \ref{def:func}.
Therefore, it is enough to specify any map $\phi \in \mathcal{S}$ on some set $I$ of generating sets $\{Y_i\}_{i \in I}$ with $Y_i \subset X_F$ by explicitly giving $\phi(Y_i)$ for all $i \in I$ and by defining
\begin{equation*}
\phi(Y) =  \bigcup\limits_{Y_i \subset Y} \phi(Y_i)
\end{equation*}
for an arbitrary $Y \subset X_F$.
The $\{Y_i\}_{i \in I}$ and $\{\phi(Y_i)\}_{i \in I}$ are thereby required to satisfy the condition (\ref{eq:generators}).
Usually, the generators $\{Y_i\}_{i \in I}$ will be taken as the sets of non-food reactants involved in a given function $\phi \in \mathcal{S}$.
This is a convenient notational simplification as the state space $\mathfrak{X}$ grows exponentially with the number of chemicals in the network.
Finally, the following notation for constant maps is used throughout the text:
\begin{equation*}
    c_Y : \mathfrak{X} \rightarrow \mathfrak{X}
\end{equation*}
denotes the map $c_Y(Z) = Y$ for all $Z \subset X_F$.
\end{rmk}

\begin{ex} \label{ex:SFood}
As an example, the semigroup model $\mathcal{S}$ for the CRS shown in Fig. \ref{fig:example3} is constructed.
\begin{figure}[htb]
  \centering
  \includegraphics[scale=0.25]{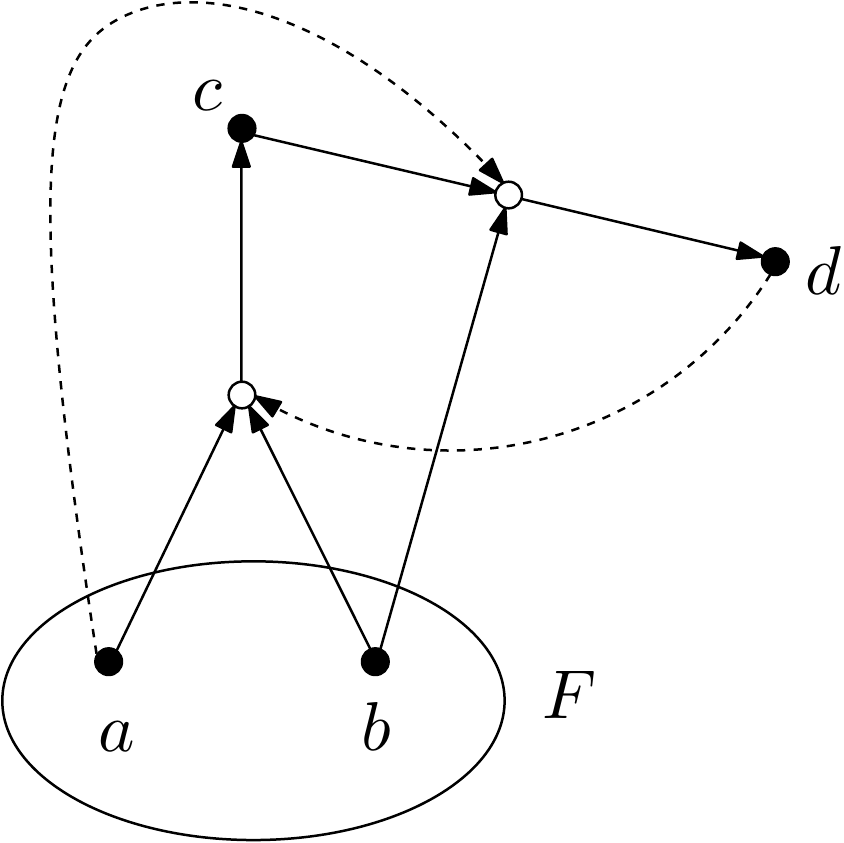}
  \caption{The CRS used to construct the semigroup model $\mathcal{S}$ in Example \ref{ex:SFood}.}
  \label{fig:example3}
\end{figure}

\noindent This model is generated by the functions $\phi_a$ and $\phi_d$, whereby $\phi_a$ has the generating sets $\emptyset$ and $\{c\}$ with $\phi_a(\emptyset) = \emptyset$ and $\phi_a(\{c\}) = \{d\}$ and $\phi_d$ is the constant function $c_{\{c\}}$.
Therefore, the sum $\phi_a + \phi_d$ has two generating sets with $(\phi_a + \phi_d)(\emptyset) = \{c\}$ and $(\phi_a + \phi_d)(\{c\}) = \{c,d\}$.
Two additional functions are generated by the products $\phi_a \circ \phi_d = c_{\{d\}}$ and $(\phi_a + \phi_d)^2 = c_{\{c,d\}}$.
One checks that these functions, together with $0$ are closed under addition and multiplication.
Thus the semigroup model $\mathcal{S}$ is the set
\begin{equation*}
\mathcal{S} = \{0,\phi_a,\phi_d,\phi_a + \phi_d,c_{\{d\}},c_{\{c,d\}}\}
\end{equation*}
equipped with the operations $\circ$ and $+$ given in Tables \ref{table:circ} and \ref{table:+}.
\begin{table} [ht]
    \caption[The operation $\circ$ in $\mathcal{S}$]{The multiplication table for $\mathcal{S}$.
    The order of composition is {\it row} $\circ$ {\it column}.}
    \label{table:circ}
\centering
    \begin{tabular}{|c|ccccc|} 
    \hline
    $\circ$ & $\phi_a$ & $\phi_d$ & $\phi_a + \phi_d$ & $c_{\{d\}}$ & $c_{\{c,d\}}$ \\
	\hline
	$\phi_a$ & 0 & $c_{\{d\}}$ & $c_{\{d\}}$ & 0 & $c_{\{d\}}$\\
	$\phi_d$ & $\phi_d$ & $\phi_d$ & $\phi_d$ & $\phi_d$  & $\phi_d$\\
	$\phi_a + \phi_d$ & $\phi_d$ & $c_{\{c,d\}}$ &  $c_{\{c,d\}}$ & $c_{\{c\}}$  & $c_{\{c,d\}}$\\
	$c_{\{d\}}$ & $c_{\{d\}}$ & $c_{\{d\}}$ & $c_{\{d\}}$ & $c_{\{d\}}$ & $c_{\{d\}}$ \\
	$c_{\{c,d\}}$ & $c_{\{c,d\}}$ & $c_{\{c,d\}}$ & $c_{\{c,d\}}$ & $c_{\{c,d\}}$ & $c_{\{c,d\}}$ \\
	\hline
	\end{tabular}
\end{table}
\begin{table} [ht]
    \caption[The operation $+$ in $\mathcal{S}$]{The addition table for $\mathcal{S}$.
    All functions $\phi$ satisfy $\phi + \phi = \phi$ giving the corresponding elements on the diagonal.
    The commutativity of addition yields the lower left half of the table.}
    \label{table:+}
\centering
    \begin{tabular}{|c|ccccc|} 
    \hline
    + & $\phi_a$ & $\phi_d$ & $\phi_a + \phi_d$ & $c_{\{d\}}$ & $c_{\{c,d\}}$ \\
   	\hline
	$\phi_a$ & & $\phi_a + \phi_d$ & $\phi_a + \phi_d$ & $c_{\{d\}}$ & $c_{\{c,d\}}$\\
	$\phi_d$ & & & $\phi_a + \phi_d$ & $c_{\{c,d\}}$ & $c_{\{c,d\}}$\\
	$\phi_a + \phi_d$ & & & & $c_{\{c,d\}}$ & $c_{\{c,d\}}$\\
	$c_{\{d\}}$ & & & & & $c_{\{c,d\}}$ \\
	$c_{\{c,d\}}$ & & & & & \\
	\hline
	\end{tabular}
\end{table}
\end{ex}

\subsection{Elementary properties} \label{sec:SElementaryProperties}

The interplay of the two operations on a semigroup model $\mathcal{S}$, the partial order on $\mathcal{S}$ and the partial order on the state space $\mathfrak{X}$ are all mutually compatible and the respective elementary properties are established in this section.

\begin{rmk} \label{rmk:PO}
 The partial order on $\mathcal{S}$ is given by $\phi \leq \psi \Leftrightarrow {\phi(Y) \subset \psi(Y)} \text{ for all $Y \subset X_F$}$ for $\phi, \psi \in \mathcal{S}$.
 For $\mathcal{S}$ endowed with this partial order, the notation $(\mathcal{S},\leq)$ is used.
\end{rmk}

\begin{lem} \label{lemma:PO}
Let $\mathcal{S}$ be a semigroup model of a CRS.
The partial order defined above is preserved under the operations $\circ$ and $+$, i.e. for any $\phi, \phi', \psi, \psi' \in \mathcal{S}$ the following relations hold

\begin{align}
\label{eq:multCompatibility}
 \phi \leq \psi \textrm{ and } \phi' \leq \psi' &\Rightarrow \phi \circ \phi' \leq \psi \circ \psi',\\
\label{eq:addCompatibility}
  \phi \leq \psi \textrm{ and } \phi' \leq \psi' &\Rightarrow \phi + \phi' \leq \psi + \psi'.
\end{align}
\end{lem}

\begin{proof}
This follows directly from the relation (\ref{eq:generators}).
\end{proof}

\begin{lem} \label{lemma:POprops}
Let $\mathcal{S}$ be a semigroup model of a CRS. Then the following properties hold true:

\noindent (I) Any $\phi, \psi \in \mathcal{S}$ satisfy 
\begin{equation} \label{eq:order1}
 \phi \leq \phi + \psi.
\end{equation}
(II) Any $\phi, \phi', \psi \in \mathcal{S}$ such that $\phi \leq \psi$ and $\phi' \leq \psi$ satisfy
\begin{equation} \label{eq:order2}
 \phi + \phi' \leq \psi.
\end{equation}
\end{lem}

\begin{proof}
Both statements follow from the relation (\ref{eq:addCompatibility}).
To prove (I), one chooses $\phi' = 0$.
The statement (II) follows from the idempotence of addition with $\psi = \psi'$.
\end{proof}

\noindent The operations $\circ$ and $+$ on $\mathcal{S}$ have the following distributivity properties.

\begin{lem} \label{lemma:distr}
Let $\phi, \psi, \chi \in \mathcal{S}$.
Then the following relations hold
\begin{align}
\label{eq:distrRight}
 \phi \circ \chi + \psi \circ \chi &= (\phi + \psi) \circ \chi, \\ 
\label{eq:distrLeft}
 \chi \circ  \phi+   \chi \circ  \psi &\leq  \chi \circ (\phi + \psi).
\end{align}
\end{lem}

\begin{proof}
Using the definitions of the operations, one obtains $(\phi \circ \chi + \psi \circ \chi)(Y) = (\phi \circ \chi)(Y) \cup (\psi \circ \chi)(Y) = \phi(\chi(Y)) \cup \psi(\chi(Y)) = (\phi + \psi)(\chi(Y)) = ((\phi + \psi) \circ \chi)(Y)$ for all $Y \subset X$, proving the equality (\ref{eq:distrRight}).

Lemma \ref{lemma:PO} and Lemma \ref{lemma:POprops}(I) imply $ \chi \circ \phi \leq \chi \circ (\phi + \psi)$ and $\chi \circ \psi \leq \chi \circ (\phi + \psi)$. 
The relation (\ref{eq:distrLeft}) then follows from Lemma \ref{lemma:POprops}(II).
\end{proof}

\begin{rmk}
The inequality in (\ref{eq:distrLeft}) can be strict.
In the example shown in Fig. \ref{fig:example4}, the function $\phi_a \circ (\phi_e + \phi_f)$ is non-zero, while $\phi_a \circ \phi_e + \phi_a \circ \phi_f$ is the zero function.
Interestingly, the property (\ref{eq:distrLeft}) reads: {\it ''The result of applying a test function $\chi$ to the sum of two functions $\phi$ and $\psi$ can be larger than the result of applying the test function to the individual functions and then taking the sum.''}
This is reminiscent of a characterization of emergence, which is often stated as {\it ''the whole is larger than the sum of its parts''}.
 \begin{figure}[htb]
  \centering
  \includegraphics[scale=0.25]{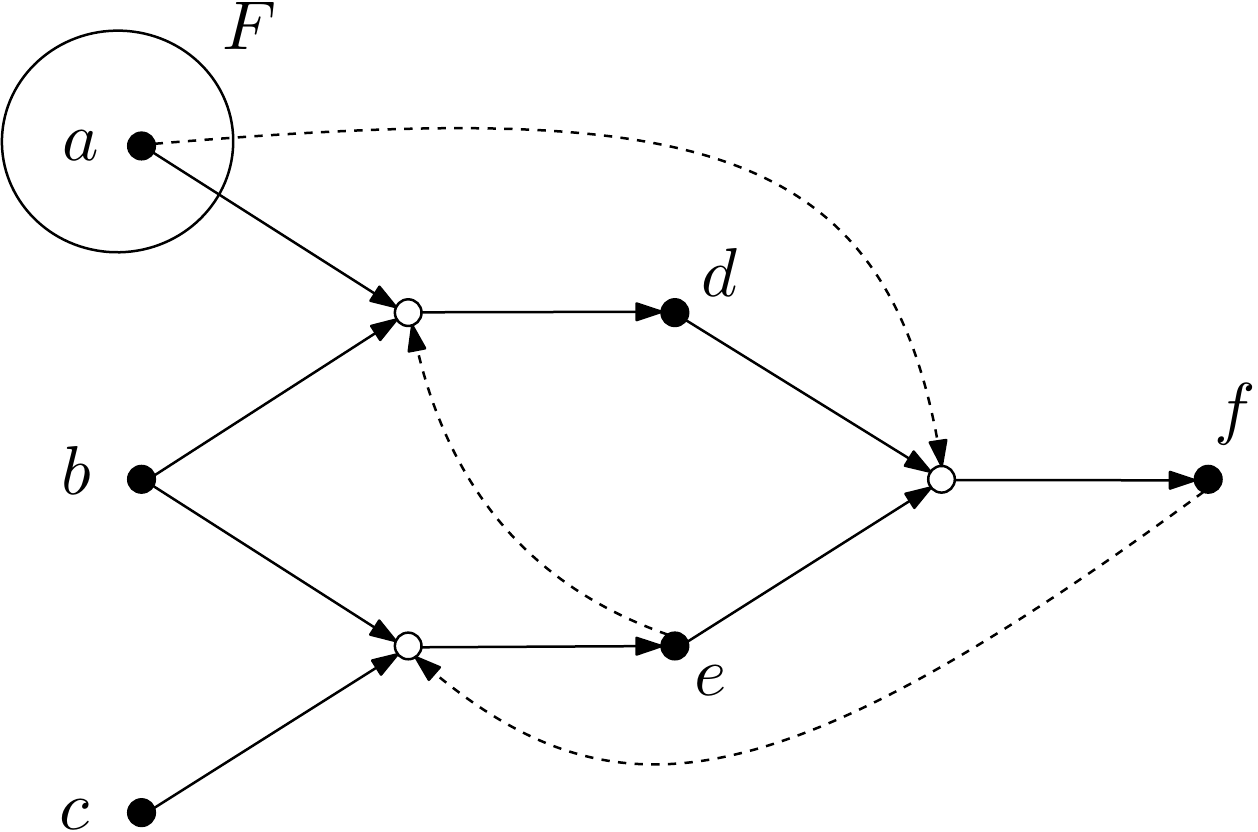}
  \caption{The semigroup model of the shown CRS has the property $\phi_a \circ (\phi_e + \phi_f) > \phi_a \circ \phi_e + \phi_a \circ \phi_f$.}
  \label{fig:example4}
\end{figure}
\end{rmk}

By definition, the semigroup $\mathcal{S}$ captures all possibilities of joint and subsequent functions of chemicals on the chemical state space $\mathfrak{X}$.
In particular, this allows to determine the functions of arbitrary subsets $Y \subset X_F$ as follows.

\begin{defn} \label{def:function}
For any $Y \subset X_F$, the semigroup model $\mathcal{S}(Y) < \mathcal{S}$ of the functions of $Y$ is defined as\footnote{Due to the constant presence of the food set, the functions of the chemicals in $F$ are included in the set of generators for $\mathcal{S}(Y)$.}

\begin{equation*}
\mathcal{S}(Y) = \langle \phi_x \rangle_{x \in Y \cup F}
\end{equation*}

\noindent and the {\it function} $\Phi_Y$ of $Y$ is given by

\begin{equation*} \label{eq:PhiY}
\Phi_{Y} = \sum_{\phi \in \mathcal{S}(Y)} \phi.
\end{equation*}

\end{defn}

The subsemigroup model $\mathcal{S}(Y)$ encodes all possible catalytic fuctionality that is supported on the set $Y \cup F \subset X$.
The function $\Phi_Y$ is characterized by the following property.

\begin{prop} \label{prop:function}
$\Phi_Y$ is the unique maximal element of $\mathcal{S}(Y)$.
\end{prop}

\begin{proof}
It suffices to show that any element $\psi \in \mathcal{S}(Y)$ satisfies $\psi \leq \Phi_{Y}$.
But this is a direct consequence of Lemma \ref{lemma:POprops}(I) as $\Phi_{Y} = \psi + \sum_{\phi \in \mathcal{S}(Y) \setminus \{\psi\}} \phi$ by construction.
\end{proof}

\begin{rmk} \label{rmk:POtranssitivity}
 If $Z \subset Y \subset X_F$, then Lemma \ref{lemma:POprops}(I) implies $\Phi_Z \leq \Phi_Y$.
 In particular, $\mathcal{S}$ has a unique maximal element, given by $\Phi_{X_F}$.
\end{rmk}

\begin{rmk} \label{rmk:subsemigroup}
Any subCRS $(X',R',C',F)$ has a semigroup model given by Definition \ref{def:semigroupModel}.
This model is denoted by $\mathcal{S}'(X',R')$.
It is a subsemigroup of the full transformation semigroup $\mathcal{T}(\mathcal{P}(X'_F))$ on the power set of $X'_F$.
A function $\phi \in \mathcal{S}'(X',R')$ can be extended to a function $ext(\phi) \in \mathcal{T}(\mathfrak{X})$ by defining 
\begin{equation*}
    ext(\phi)(Y) = \phi(Y \cap X_F')
\end{equation*}
for all $Y \subset X_F$.
This gives a homomorphic embedding of $\mathcal{S}'(X',R')$ into $\mathcal{T}(\mathfrak{X})$, but $ext(\mathcal{S}'(X',R'))$ is neither a subsemigroup of $\mathcal{S}(X'_F)$ nor of $\mathcal{S}$.

Because the generators $\{\phi'_x\}_{x \in X'} \subset \mathcal{T}(\mathcal{P}(X'_F))$ of $\mathcal{S}'(X',R')$ and the generators $\{\phi_x\}_{x \in X'} \subset \mathcal{T}(\mathfrak{X})$ of $\mathcal{S}(X'_F)$ satisfy $ext(\phi'_x) \leq \phi_x$ for all $x \in X'$, the the inequality
\begin{align} \label{eq:maxInequality}
    ext(\Phi'_{X'_F}) \leq \Phi_{X'_F}
\end{align}
for the maximal functions $\Phi'_{X'_F}$ and $\Phi_{X'_F}$ of $\mathcal{S}'(X',R')$ and $\mathcal{S}(X'_F)$ follows from Lemma \ref{lemma:PO}.
\end{rmk}

\section{Characterization of self-sustaining CRS} \label{sec:RAF}

In this section, it is illustrated how the semigroup model of a CRS provides a natural language to formulate and prove statements regarding CRS in a concise manner.
The specific focus here lies on the application to self-sustaining CRS.
In Section \ref{sec:PRAF}, it is shown that a CRS is self-sustaining if and only if its function $\Phi_{X_F}$ satisfies $\Phi_{X_F}(X_F) = X_F$ (Theorem \ref{thm:PRAF}) and it follows that any self-sustatining set of chemicals $X_F' \subset X_F$ satisfies $X'_F \subset \Phi_{X'_F}(X'_F)$ (Corollary \ref{corr:subPRAF}).
Moreover, the condition $X'_F = \Phi_{X'_F}(X'_F)$ is sufficient for $X'_F$ to be a self-sustaining set of chemicals (Proposition \ref{prop:PRAFsuff}).
This condition is, however, not necessary.
In Section \ref{sec:dynamics}, a discrete dynamics is introduced by considering the repeated action of a set of chemicals, and the elementary properties of the dynamics are established.
In Section \ref{sec:maxsubPRAF} it is proven that the maximal self-sustaining set of chemicals of a CRS is the fixed point of the dynamics with the initial condition given by the presence of all chemicals of the CRS (Theorem \ref{thm:maxPRAF}).
As a corollary, it follows that a CRS with a nilpotent semigroup model cannot have nontrival self-sustaining subCRS.
This results establishes a link between the combinatorial theory of finite semigroups and the the theory of self-sustaining chemical reaction networks.
Finally, a self-sustaining set of chemicals $X_F'$ is in general not stable under the dynamics but converges to a fixed point $X_F'^{*s}$ which contains $X_F'$.
This fixed point is a self-sustaining set of chemicals itself and is termed the {\it functional closure} of $X_F'$.
The importance of functionally closed sets of chemicals is discussed and a structure theorem on the lattice of all functionally closed sets of chemicals is proven (Theorem \ref{thm:M}).

Throughout this section, fix a CRS  $(X,R,C,F)$ and let $\mathcal{S}$ be its semigroup model.

\subsection{Characterization of self-sustaining CRS} \label{sec:PRAF}

Recall that the self-sustainment property of a CRS in Definition \ref{def:PRAF} requires all chemicals in $X_F$ to be produced from chemicals in $X$ by catalyzed reactions.
Putting this description into mathematical language, one would expect a self-sustaining CRS to satisfy $\Phi_{X_F}(X_F) = X_F$.
This is precisely the statement of the following theorem.

\begin{thm} \label{thm:PRAF}
 A CRS is self-sustaining if and only if its maximal function $\Phi_{X_F}$ satisfies
 \begin{equation*}
     \Phi_{X_F}(X_F) = X_F.
 \end{equation*}
\end{thm}

\begin{proof}
If the CRS is self-sustaining, then there must be a set of reactions $R' \subset \pi_R(C)$ such that $\textrm{ran}(R') = X_F$.
For each reaction $r \in R'$, choose a catalyst $y(r) \in X$.
Then the function
\begin{equation*}
    \phi:= \sum_{r \in R'} \phi_{y(r)}
\end{equation*}
is an element of $\mathcal{S}$ and satisfies $\phi(X_F) = X_F$.
The maximal function $\Phi_{X_F}$ is bounded from below by $\phi$ and therefore also satisfies $\Phi_{X_F}(X_F) = X_F$.

To prove the reverse, assume that $\Phi_{X_F}(X_F) = X_F$ holds.
The function $\Phi_{X_F}$ must be of the form
\begin{equation*}
    \Phi_{X_F} = \sum_{y \in Y} \phi_y \circ \psi_y
\end{equation*}
for some subset of generating functions $\{\phi_y\}_{y \in Y}$ with $Y \subset X$ and for some functions $\psi_y \in \mathcal{S}$\footnote{Note that, {\it a priori}, $Y$ should be a multiset with elements in $X$ and the argument of the proof would go through unaltered if $Y$ was a multiset.
However, using the relation (\ref{eq:distrLeft}), i.e. $\phi_y \circ \psi_y + \phi_y \circ \psi'_y \leq \phi_y \circ (\psi_y + \psi'_y)$, and the maximality of $\Phi_{X_F}$, one sees that the given representation of $\Phi_{X_F}$ with a set $Y$ is actually correct.}.
The function
\begin{equation*}
    \psi := \sum_{y \in Y} \phi_y    
\end{equation*}
satisfies $\psi(X_F) = X_F$, as otherwise $\Phi_{X_F}$ could not have $X_F$ in its image.
Thus, the set of reactions
\begin{equation*}
    R' = \{ r \in R \textrm{ such that }(y,r) \in C \textrm{ for some }y \in Y\}
\end{equation*}
satisfies the condition \ref{cond:PR} as the condition $\textrm{dom}(R') \subset \textrm{ran}(R') \cup F = X$ is trivially satisfied.
\end{proof}

\begin{corr} \label{corr:subPRAF}
 If $X'_F \subset X_F$ is a self-sustaining set of chemicals, then the inclusion $X'_F \subset \Phi_{X'_F}(X'_F)$ holds.
\end{corr}
\begin{proof}
Let $\Phi'_{X'_F}$ be the maximal function of the semigroup model $\mathcal{S}'(X',R\mid_{X'})$ for the subCRS generated by $X'_F$.
By Theorem \ref{thm:PRAF}, the equality $\Phi'_{X'_F}(X'_F) = X'_F$ holds.
Moreover, let $\Phi_{X'_F}$ be the maximal function of the subsemigroup model $\mathcal{S}(X'_F) < \mathcal{S}$.
Then the inequality (\ref{eq:maxInequality}) yields the desired inclusion
\begin{align*}
    X'_F = ext(\Phi'_{X'_F})(X'_F) \subset \Phi_{X'_F}(X'_F).
\end{align*}
\end{proof}

The inclusion $X'_F \subset \Phi_{X'_F}(X'_F)$ proven in the above corollary can be strict in the case that $X'$ contains catalysts that catalyze reactions with products outside of $X'$ but with reactants in $X'$.
Therefore, the equality $X'_F = \Phi_{X'_F}(X'_F)$ cannot be a necessary condition for self-sustaining sets of chemicals.
It is, however, a sufficient one:

\begin{prop} \label{prop:PRAFsuff}
If a set of chemicals $X'_F \subset X_F$ satisfies the equality $X'_F = \Phi_{X'_F}(X'_F)$, then it is a self-sustaining set of chemicals.
\end{prop}

\begin{proof}
The function $\Phi_{X'_F} \in \mathcal{S}(X'_F)$ is of the form
\begin{align*}
    \Phi_{X'_F} &= \sum_{y \in Y} \phi_y \circ \psi_y \\
        &= \sum_{y \in Y} \sum_{(y,r) \in C} \phi_r \circ \psi_y 
\end{align*}
for some subset of generating functions $\{\phi_y\}_{y \in Y}$ with $Y \subset X'$ and for some ${\psi_y \in \mathcal{S}(X'_F)}$.
Define the function $\psi := \sum_{y \in Y} \psi_y$.
The function $\Phi_{X'_F}$ is bounded above by $ (\sum_{y \in Y} \phi_y) \circ \psi$, but due to its maximality in $\mathcal{S}(X'_F)$ it must be equal to this function, i.e. $\Phi_{X'_F} = (\sum_{y \in Y} \phi_y) \circ \psi$ with $\psi \in \mathcal{S}(X'_F)$, or, equivalently
\begin{equation} \label{eq:PhiXF}
    \Phi_{X'_F} = \sum_{y \in Y} \sum_{(y,r) \in C} \phi_r \circ \psi. 
\end{equation}
Moreover, the relation $\psi \leq \Phi_{X'_F}$ holds, again due to the maximality of $\Phi_{X'_F}$, and this implies $\psi(X'_F) \subset X'_F$.
This means that in the calculation of $\Phi_{X'_F}(X'_F)$ only those $\phi_r$ in the representation (\ref{eq:PhiXF}) with $\textrm{dom}(r) \subset \psi(X'_F) \cup F \subset X'$ play a role.
In other words, choosing the set of reactions $R'$ as
\begin{equation*}
    R' := \left\{ r \in R \textrm{ such that } (y,r) \in C \textrm{ for some }y \in Y \textrm{ and }\textrm{dom}(r) \subset X' \right\},
\end{equation*}
one can write $X'_F = \Phi_{X'_F}(X'_F)$ as
\begin{equation}
    X'_F = \Phi_{X'_F}(X'_F) = \left(\sum_{r \in R'} \phi_r \circ \psi\right)\left(X'_F\right) = \bigcup_{r \in R'} \phi_r(\psi(X'_F)).
\end{equation}
This means that $X'_F = \textrm{ran}(R')$.
Additionaly, $\textrm{dom}(R') \subset X'$ holds by definition, and thus the set $R'$ satisfies the condition \ref{cond:PR}.
As $R' \subset R \mid_{X'}$, this implies that $X'_F$ is a self-sustaining set of chemicals.
\end{proof}

\begin{rmk}
Although the statement of Proposition \ref{prop:PRAFsuff} appears trivial at first glance, one needs to take care of the fact that $\Phi_{X'_F}$ is a function on the power set of $X_F$ and not on the power set of $X'_F$.
Therefore, it could potentially include the formation of chemicals in $X'_F$ via reaction pathways which include intermediate chemicals outside of $X'$.
The essence of the proof is to show, based on the maximality of $\Phi_{X'_F}$ in $\mathcal{S}(X'_F)$, that such pathways do not exist under the condition $X'_F = \Phi_{X'_F}(X'_F)$.
\end{rmk}

\subsection{Sustaining dynamics on a Semigroup Model} \label{sec:dynamics}

There is natural discrete dynamics on a CRS based on its semigroup model.
Starting with any set of chemicals $Y_0 \subset X_F$, the a maximal function $\Phi_{Y_0}$ (cf. Definition \ref{def:function}) acts on the set itself.
This yields the maximal set $Y_1=\Phi_{Y_0}(Y_0)$ that can be produced from $Y_0$ and the food set using the functionality supported on $Y_0$ and the food set.
This argument applies iteratively and gives rise to the {\it sustaining dynamics} on a CRS.

\begin{defn}
The {\it sustaining dynamics} of a CRS with the initial condition $Y_0 \subset X_F$ is generated by the propagator
\begin{align*}
\mathcal{D}^s : \mathfrak{X} &\rightarrow \mathfrak{X}\\
Y &\mapsto \Phi_Y(Y),
\end{align*}
where $\Phi_Y$ is the function of $Y \subset X_F$.
Analogously, the dynamics is parametrized by $\mathbb{Z}_{\geq 0}$ as
\begin{equation*}
Y_{n+1} = \Phi_{Y_n}(Y_n) \text{ for all } n \in \mathbb{Z}_{\geq 0}.
\end{equation*}
\end{defn}

\noindent Note that the propagator $\mathcal{D}^s$ deletes all elements that are in $Y$ but not in $\Phi_Y(Y)$.

\begin{rmk} \label{rmk:period}
Because the state space $\mathfrak{X}$ is finite, for the sequence $(Y_n)_{n \in \mathbb{Z}_{\geq 0}}$ generated by $\mathcal{D}^s$ under the initial condition $Y_0$, there exist minimal nonnegative integers $k$ and $m \neq 0$ such that $Y_k = Y_{k+m}$.
This gives rise to periodic behavior, i.e. $Y_{k+i} = Y_{k+i+nm}$ for all $i=0,...,m-1$ and all $n \in \mathbb{N}$.
If $m=1$, then $Y_k$ is a fixed point and the dynamics is said to {\it stabilize} at $Y_k$.
A fixed point which results from the initial condition $Y_0$ will be denoted by $Y_0^{*s}$.
If $m>1$, the dynamics has period $m$ and is {\it oscillatory}.
Both behaviors are possible in CRS.
\end{rmk}

\begin{rmk}
According to Theorem \ref{thm:PRAF}, if the CRS is self-sustaining, then $X_F$ is a fixed point for the dynamics with initial condition $Y_0=X_F$.
\end{rmk}

\begin{ex} \label{ex:period}
Fig. \ref{fig:circDynamics} shows a CRS with $X=\{a,b,c\}$, food set $F=\emptyset$ and the respective reactions shown in the Figure.
If the initial condition $Y_0$ is a proper subset of $X_F$, the dynamics has period 3.
For example, the dynamics with the initial condition $Y_0=\{a\}$ is
\begin{equation*}
\{a\} \mapsto \{b\} \mapsto \{c\} \mapsto \{a\} \mapsto ...
\end{equation*}
\begin{figure}[htb]
  \centering
  \includegraphics[scale=0.25]{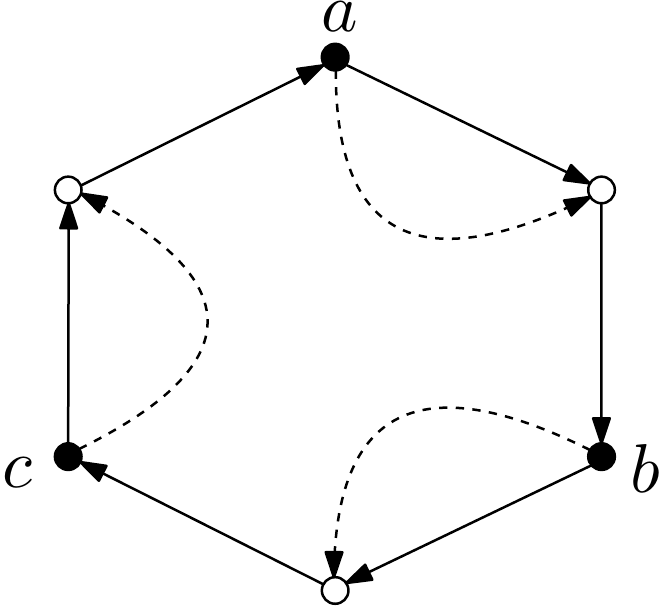}
  \caption{Example of a CRS with possible oscillatory dynamics.
  The food set is the empty set.}
  \label{fig:circDynamics}
\end{figure}
\end{ex}

The elementary properties of the sustaining dynamics are described in the following propositions.

\begin{prop}
Let $(Y_n)_{n \in \mathbb{Z}_{\geq 0}}$ be the sustaining dynamics with initial condition $Y_0$. 
If the semigroup $(\mathcal{S},\circ)$ of the CRS is nilpotent, then the dynamics stabilizes at $\emptyset$, i.e. there exists a natural number $N$ such that
\begin{equation*}
Y_n = \emptyset \text{ for all } n \geq N.
\end{equation*}
\end{prop}

\begin{proof}
By definition, $Y_n = \Phi_{Y_{n-1}} \circ \Phi_{Y_{n-2}} \circ ... \circ \Phi_{Y_0}(Y_0)$ holds.
Because $\mathcal{S}$ is nilpotent, there exists an index $N$ such that $\mathcal{S}^N = \{0\}$.
This implies that $Y_n = \emptyset$ for all $n \geq N$.
\end{proof}

A useful result is that the dynamics with initial condition $X_F$ cannot have periodic behavior, but always has a fixed point.
It is a consequence of the following more general proposition.

\begin{prop} \label{prop:fixedPoint}
Let $(Y_n)_{n \in \mathbb{Z}_{\geq 0}}$ be the sustaining dynamics such that $Y_1 \subset Y_0$.
Then
\begin{align*}
Y_{n+1} \subset Y_n
\end{align*}
holds for all $n \in \mathbb{Z}_{\geq 0}$ and the dynamics stabilizes at the fixed point $Y_0^{*s}$.
\end{prop}

\begin{proof}
The proof proceeds by induction. 
By hypothesis $Y_1 \subset Y_0$ is satisfied.
Let $Y_n \subset Y_{n-1}$.
This implies the ordering of the respective functions $\Phi_{Y_n} \leq \Phi_{Y_{n-1}}$ by Remark \ref{rmk:POtranssitivity}.
Together with Remark \ref{rmk:generators} this yields
\begin{equation*}
Y_{n+1} = \Phi_{Y_n}(Y_n) \subset \Phi_{Y_{n-1}}(Y_n) \subset \Phi_{Y_{n-1}}(Y_{n-1}) = Y_n.
\end{equation*}
Thus the dynamics is given by a descending chain of sets 
\begin{equation*}
    Y_0 \supset Y_1 \supset ... \supset Y_n \supset Y_{n+1} \dots
\end{equation*}
Because $X_F$ is finite, the chain stabilizes.
\end{proof}

\begin{corr} \label{corr:stabDyn}
A dynamics $(Y_n)_{n \in \mathbb{Z}_{\geq 0}}$ with initial condition $Y_0 = X_F$ always leads to a fixed point.
\end{corr}

\begin{proof}
This follows from $\Phi_{X_F}(X_F) \subset X_F$ and the previous proposition.
\end{proof}

\subsection{Identification of self-sustaining sets of chemicals} \label{sec:maxsubPRAF}

In this section, the main theorem concerning the maximal self-sustaining set of chemicals of a CRS is stated and its biological significance is discussed.
Its proof is merely a combination of the previously proven results.

The following lemma is required to ensure that the sustaining dynamics is well-behaved with respect to self-sustaining sets of chemicals.

\begin{lem} \label{lem:PRAFinclusion}
 If $X'_F \subset X_F$ is a self-sustaining set of chemicals and if $Y$ is a set that satisfies $X'_F \subset Y \subset X_F$, then the inclusion
 \begin{equation*}
     X'_F \subset \Phi_Y(Y)
 \end{equation*}
 holds.
\end{lem}
\begin{proof}
The chain of inclusions
\begin{equation*}
    X'_F \subset \Phi_{X'_F}(X'_F) \subset \Phi_Y(X'_F) \subset \Phi_Y(Y)
\end{equation*}
follows from the Corollary \ref{corr:subPRAF} and the Remarks \ref{rmk:generators} and \ref{rmk:POtranssitivity}.
\end{proof}

Now the main theorem can be proven:

\begin{thm}[On the maximal self-sustaining set of chemicals] \label{thm:maxPRAF}
For any CRS, the maximal self-sustaining set of chemicals is the fixed point of the sustaining dynamics $(Y_n)_{n \in \mathbb{Z}_{\geq 0}}$ with the initial condition $Y_0 = X_F$, i.e. it is the set $X_F^{*s}$.
\end{thm}

\begin{proof}
By Corollary \ref{corr:stabDyn}, the dynamics stabilizes at a the fixed point $X_F^{*s}$.
Being a fixed point, the set $X_F^{*s}$ satisfies
\begin{equation*}
    X_F^{*s} = \Phi_{X_F^{*s}}(X_F^{*s})
\end{equation*}
and is thus a self-sustaining set of chemicals by Proposition \ref{prop:PRAFsuff}.
Let $X'_F \subset X_F$  be any self-sustaining set of chemicals.
By inductively applying the Lemma \ref{lem:PRAFinclusion} to the sets $(Y_n)_{n \in \mathbb{Z}_{\geq 0}}$, it follows that $X'_F \subset X_F^{*s}$, which yields the maximality of $X_F^{*s}$.
\end{proof}

\begin{rmk}
The result of this theorem is closely related to an algorithm to determine the maximal pseudo-RAF set of reactions introduced by \cite{Hordijk2015}.
Therein, a dynamics $\{R_n\}_{n \in \mathbb{Z}_{\geq 0}}$ on the power set $\mathcal{P}(R)$ of reactions with the initial condition $R_0 = R$ is used.
The propagator $R_n \mapsto R_{n+1}$ is defined by the operation
\begin{align*}
    R_{n+1} = \{ r \in R_n \textrm{ such that $r$ has a catalyst in }\textrm{ran}(R_{n}) \cup F\}.
\end{align*}
This is equivalent to the dynamics $Y \mapsto \Phi_Y(Y)$ with the initial condition ${Y_0 = X_F}$ via the mutually inverse mapping $Y_n = \textrm{supp}(R_n)$ and ${R_n = \pi_{R\mid_{Y_n \cup F}}(C\mid_{Y_n \cup F})}$.
\end{rmk}

\noindent Theorem \ref{thm:maxPRAF} gives to the following corollary.

\begin{corr} \label{corr:nilpotentPRAF}
A CRS with a nilpotent semigroup $(\mathcal{S},\circ)$ has no nontrivial self-sustaining subCRS.
\end{corr}

\begin{proof}
If there was a nontrivial self-sustaining subCRS and thus a nonempty self-sustaining set of chemicals, there would be a maximal self-sustaining set of chemicals $X_F^{*s} \neq \emptyset$, which satisfies $X_F^{*s} = \Phi_{X_F^{*s}}(X_F^{*s})$.
Any power $\Phi_{X_F^{*s}}^n$ of the function $\Phi_{X_F^{*s}}$ applied to $X_F^{*s}$ yields $X_F^{*s}$ and is therefore nonzero.
\end{proof}

This corollary has importance when studying the semigroups that arise from self-sustaining CRS or CRS with nontrivial self-sustaining subCRS as it weeds out all nilpotent semigroups $(\mathcal{S},\circ)$, which is the largest class of finite semigroups, cf. \cite{Satoh1994,Almeida1995}, as potential semigroup models.

\subsection{Functionally closed self-sustaining sets of chemicals} \label{sec:functionalClosure}

In the Theorem \ref{thm:maxPRAF}, it was shown how to algorithmically identify the maximal self-sustaining set of chemicals as the fixed point $X_F^{*s}$ of the sustaining dynamics.
For any given CRS, the lattice of self-sustaining subCRS is of importance to gain a deeper understanding of the modular structure the system.
However, not all self-sustaining subCRS carry the same biological importance.
For example, subCRS which are not closed according to Definition \ref{def:CRAF} should be considered as purely theoretical constructs because any subCRS occurring in the biophysical reality will always be closed.
But the closure property as given in the Definition \ref{def:CRAF} only accounts for the closure on the level of reactions and not on the level of the functions of chemicals.
This leads to the possibility that for a self-sustaining set of chemicals $X'_F \subset X_F$, the inclusion $X'_F \subset \Phi_{X'_F}(X'_F)$ is strict, i.e. that there are elements in $X_F \setminus X'_F$ which are produced from $X'$ by functionality supported on $X'$.
Therefore, a self-sustaining set $X'_F$ which is strictly included in $\Phi_{X'_F}(X'_F)$ will not be stable but will produce the chemicals $\Phi_{X'_F}(X'_F)$ over time.
Then the chemicals in $\Phi_{X'_F}(X'_F)$ might catalyze the production of even more chemicals from $\Phi_{X'_F}(X'_F)$ and the food set.
The following analogue of Proposition \ref{prop:fixedPoint} ensures that this dynamics will lead to a fixed point.

\begin{prop} \label{prop:increasingFixedPoint}
Let $(Y_n)_{n \in \mathbb{Z}_{\geq 0}}$ be the sustaining dynamics such that $Y_0 \subset Y_1$.
Then $Y_{n} \subset Y_{n+1}$ holds for all $n \in \mathbb{Z}_{\geq 0}$ and the dynamics stabilizes at the fixed point $Y_0^{*s}$.
\end{prop}
\begin{proof}
Analogous to the proof of Proposition \ref{prop:fixedPoint}.
\end{proof}

Once the fixed point $X_F'^{*s}$ is reached, no further chemicals will be produced from the set $X_F'^{*s} \cup F$ by the functionality supported on it.
By Proposition \ref{prop:PRAFsuff}, the fixed point $X_F'^{*s}$ is a self-sustaining set of chemicals.
Because $X_F'^{*s}$ contains the original self-sustaining set of chemicals $X'_F$, it should be thought of as the functional closure of $X'_F$.
This suggests the following definition.

\begin{defn} \label{def:funcClosed}
For a self-sustaining set of chemicals $X'_F$, the fixed point of the sustaining dynamics $X_F'^{*s} \supset X'_F$ is called the {\it functional closure} of $X'_F$.
A set of chemicals $X'_F \subset X_F$ is {\it functionally closed} if it satisfies
\begin{equation*}
    \Phi_{X'_F}(X'_F) = X'_F.
\end{equation*}
The corresponding subCRS $(X',R\mid_{X'},C\mid_{X'},F)$ is said to be a {\it functionally closed subCRS}\footnote{Note that a functionally closed subCRS is always closed in the sense of Definition \ref{def:CRAF}.}.
\end{defn}

This definition has not been proposed in the CRS literature thus far to the best of the author's knowledge.
One reason might be that it is cumbersome to formulate in the standard CRS terminology.
However, it is a natural definition when the language of semigroup models is available.
The functionally closed sets of chemicals of any given CRS are important objects to understand the modular structure of the CRS, because they correspond to core functional modules of the CRS.
Note that due to Proposition \ref{prop:PRAFsuff}, all functionally closed sets of chemicals are self-sustaining.

\begin{ex}
Consider the CRS shown in Fig. \ref{fig:example1}.
The subCRS generated by $X'_F = \{c,d\}$ is self-sustaining, but not functionally closed as $\Phi_{X'_F}(X'_F) = X_F$ yields the set of all non-food chemicals, which is functionally closed.
In fact, $X_F$ is the only nontrivial functionally closed set of chemicals of the given CRS.
\end{ex}

It is possible to algorithmically determine the hierarchy of all functionally closed sets of chemicals of any given CRS.
This is based on the following modification of the sustaining dynamics.

\begin{defn}
The {\it reduced sustaining dynamics} of a CRS with the initial condition $Y_0 \subset X_F$ is generated by the propagator
\begin{align*}
\mathcal{D}^{rs} : \mathfrak{X} &\rightarrow \mathfrak{X}\\
Y &\mapsto Y \cap \Phi_Y(Y).
\end{align*}
The trajectory $\{Y_n\}_{n \in \mathbb{Z}_{\geq 0}}$ of the dynamics is parametrized by $\mathbb{Z}_{\geq 0}$ as
\begin{equation*}
Y_{n+1} = Y_n \cap \Phi_{Y_n}(Y_n) \text{ for all } n \in \mathbb{Z}_{\geq 0}.
\end{equation*}
\end{defn}

\begin{rmk}
Note that for the initial condition $Y_0 = X_F$, the reduced sustaining dynamics yields the same trajectory as the sustaining dynamics, because $\Phi_{Y_n}(Y_n) \subset Y_n$ is guaranteed by Proposition \ref{prop:fixedPoint}.
\end{rmk}

The reduced sustaining dynamics behaves well with respect to functionally closed sets of chemicals:

\begin{lem} \label{lem:fcSubsets}
The reduced sustaining dynamics stabilizes for any initial condition $Y_0 \subset X_F$.
Its fixed point, which is denoted by $Y_0^{*rs}$, is the maximal functionally closed set of chemicals that is contained in $Y_0$.
\end{lem}
\begin{proof}
By definition, the trajectory of the dynamics is given by a descending chain of sets $Y_0 \supset Y_1 \supset \dots \supset Y_n \supset Y_{n+1} \dotsm$ and therefore it stabilizes.
The fixed point satisfies the equation
\begin{align*}
    Y_0^{*rs} = Y_0^{*rs} \cap \Phi_{Y_0^{*rs}}(Y_0^{*rs}),
\end{align*}
which is equivalent to the fixed point equation $Y_0^{*rs} = \Phi_{Y_0^{*rs}}(Y_0^{*rs})$, i.e. $Y_0^{*rs}$ is functionally closed.

Let $X'_F \subset Y_0$ be a functionally closed (and therefore self-sustaining) set of chemicals contained in $Y_0$
By Lemma \ref{lem:PRAFinclusion} it is contained in $\Phi_{Y_0}(Y_0)$ and thus in $Y_0 \cap \Phi_{Y_0}(Y_0)$.
An inductive application of this argument proves that $X'_F \subset Y_n$ for all $n \in \mathbb{Z}_{\geq 0}$ and thus $X'_F \subset Y_0^{*rs}$, which shows the maximality of $Y_0^{*rs}$.
\end{proof}

\noindent This suggests to define, for any set $Y \subset X_F$, the set
\begin{equation*}
   \mathfrak{M}(Y) := \left\{  (Y\setminus\{y\})^{*rs} \right\}_{y \in Y} \subset \mathfrak{X}
\end{equation*}
and, furthermore, to define iteratively 
\begin{align*}
    \mathfrak{M}^0 &:= \{ X_F^{*s} \} \\
    \mathfrak{M}^{i+1} &:= \bigcup_{Y \in  \mathfrak{M}^{i}} \mathfrak{M}(Y)
\end{align*}
for $i \in \mathbb{Z}_{\geq 0}$.
This construction yields an (algorithmic) description of all functionally closed sets of chemicals for any given CRN:

\begin{thm}[On the functionally closed subsets of chemicals] \label{thm:M}
There exists an $N \in \mathbb{Z}_{\geq 0}$ such that 
\begin{equation*}
    \mathfrak{M}^{i} = \{ \emptyset \}
\end{equation*}
for all $i >N$.
Moreover, the set
\begin{equation*}
    \mathfrak{M} := \bigcup_{i=0}^N \mathfrak{M}^{i} \subset \mathfrak{X}
\end{equation*}
is the set of all functionally closed sets of chemicals.
\end{thm}
\begin{proof}
For each $Y \in \mathfrak{M}^{i+1}$, there exists a set in $\mathfrak{M}^{i}$ which strictly contains $Y$ by construction .
Therefore, choosing $N := \mid X_F^{*s} \mid$ proves the first statement.

The maximal functionally closed set of chemicals $X_F^{*s}$ is contained in $\mathfrak{M}$ by construction.
Let $Y \subset X_F^{*s}$ be any functionally closed set of chemicals.
Then there exists a sequence of maximal length, which consists of functionally closed sets
\begin{equation*}
    Y = Y_n \subsetneq Y_{n-1} \subsetneq \dots \subsetneq Y_1 \subsetneq X_F^{*s}.
\end{equation*}
Then $Y_i \in \mathfrak{M}^{i}$ by construction and thus $Y_n \in \mathfrak{M}^{n}$ for some $n \leq N$.
Therefore, $\mathfrak{M}$ contains all functionally closed sets of chemicals.
The reverse inclusion follows from Theorem \ref{thm:maxPRAF} and Lemma \ref{lem:fcSubsets}, which state that all elements of $\mathfrak{M}$ are functionally closed sets of chemicals.
\end{proof}

The set $\mathfrak{M}$ can be constructed algorithmically and the Theorem \ref{thm:M} is thus applicable to the analysis of the CRS of real biological systems.
Moreover, only the knowledge of the functions $\Phi_Y$ of sets of chemicals $Y \subset X_F$ is required to determine both $X_F^{*s}$ and $\mathfrak{M}$.
The knowledge of the full semigroup model is, on the contrary, not required.\\

Functionally closed sets of chemicals should be important to gain insight into the modular structure of self-sustaining CRS.
Moreover, this notion could be useful to analyze the evolutionary aspects of self-sustaining CRS.
For example, if for a given chemical, there exists a unique minimal functionally closed set of chemicals containing it, then the chemical should be considered as a part of the module represented by the respective functionally closed set.
If, however, there are multiple minimal functionally closed sets containing this chemical, then it should be considered as a mediator between these sets and would more likely have appeared when the respective modules were combined.

\section{Discussion} \label{sec:discussion}

In this work, it was shown that the CRS formalism has a natural algebraic structure induced by the simultaneous and subsequent functions of catalysts.
The constructed semigroups contain all possible functional combinations and thus faithfully reflect the catalytic properties of the network.
Moreover, the partial order on $\mathcal{S}$ allows to assign a well-defined catalytic function to any subset of chemicals.
These functions respect the partial order on subsets of chemicals via inclusion and allow to define a sustaining and a reduced sustaining dynamics of the CRS.
Finally, the interplay of these structures yields a characterization of the maximal self-sustaining set of chemicals.
Moreover, a closer inspection of the sustaining dynamics naturally leads to the definition of functionally closed sets of chemicals.
Finally, a combination of the methods developed in this article allows to give a characterization of the lattice of all functionally closed sets of chemicals.
Such sets of chemicals should play an important role in the analysis of the modular structure of CRS and pose valuable targets for the analysis of CRS corresponding to real biological systems, such as the ones constructed by \cite{Sousa2015} and \cite{Xavier2022}.\\

Self-sustaining chemical reaction networks are an important area of research and this work introduces new techniques to field, which are based on the potentially powerful methods from semigroup theory.
For example, Corollary \ref{corr:nilpotentPRAF} shows that a CRS with nilpotent semigroup cannot contain any nontrivial self-sustaining subCRS.
This is an important fact by itself because most semigroups are nilpotent (any magma\footnote{A magma is a semigroup without the associativity property.} with the product of any three elements equal to zero is automatically a semigroup) and this weeds out these objects in the study of self-sustaining networks.
In semigroup theory, combinatorial problems are an important and developed field and with the construction provided in this article, such methods can now be applied to the combinatorics of self-sustaining CRS.
However, all central results given here require a representation of the semigroup model as a subsemigroup of the full transformation semigroup $\mathcal{T}(\mathfrak{X})$ and are based on particular properties of this representation.
It would be interesting to find purely algebraic descriptions of semigroup models of CRS and to restate the main results purely algebraically, i.e. with without reference to this representation.\\

Moreover, there is an equivalence between finite semigroups and finite automata established by \cite{Schutzenberger1965}.
Thus the semigroup models developed here suggest to investigate the computational capabilities of catalytic reaction systems as a future direction of research.
In this regard, it is interesting to study the inverse problem, i.e. to determine which finite semigroups can be realized as semigroup models of CRS and to analyze their computational properties. 
It is clear that not all semigroups can be interpreted as semigroup models of some CRS, because for a general finite semigroup, a partial order satisfying Lemma \ref{lemma:PO} does not exist.\\

Further interesting questions in this direction arise for the semigroups of infinite reaction networks and the classification of their computational properties.
The definitions given here extend directly to infinite networks, but the arguments based on finiteness of $\mathcal{S}$ and $\mathfrak{X}$ used in many proofs then require modification and in some cases the analogous results do not hold.
For example, the discrete dynamics can lead to a steadily growing network instead of a fixed point or periodic orbit.
Another possibility is the extension of the state space from $\{0,1\}^{X_F}$ to $\mathbb{R}_{\geq 0}^{X_F}$ by taking into account the concentrations of the chemicals.
In this case, the dynamics is governed by the classical kinetic rate equations and following the line of thought presented in this work could lead to the development of a notion of function and causality for such classical models.\\

The main motivation for the construction of algebraic models for chemical reaction networks is a natural possibility of coarse-graining through taking quotients by congruences.
In this approach, the possible coarse-graining procedures are given by the lattice of congruences on the algebraic structure.
This has the advantage that the coarse-grained model is naturally equipped with the same structure as the original model and that consecutive coarse-graining procedures over several scales are feasible.
This approach will be investigated in detail in a forthcoming publication.\\

Finally note that self-generating CRS, which are called RAF in the literature, can be characterized in an analogous fashion as presented for self-sustaining CRS in this work.
The only difference is that the sustaining dynamics $Y \mapsto \Phi_Y(Y)$ needs to be replaced by the generative dynamics $Y \mapsto \Phi_Y(\emptyset)$.
However, the proofs require a deeper understanding of the structure of the semigroups, which is provided in a companion article by \cite{loutchko2022arxiv}.

\backmatter

\bmhead{Acknowledgments}

I am deeply indebted to Gerhard Ertl for valuable discussions and his enormous support.
I thank Tetsuya J. Kobayashi, all members of the Kobayashi lab, and Hiroshi Kori for stimulating discussions and their support.
I also thank the anonymous reviewers for the careful reading of the first version of the manuscript and their various helpful suggestions.
This research is supported by JSPS KAKENHI Grant Numbers 19H05799 and 21K21308, and by JST CREST JPMJCR2011 and JPMJCR1927.


\bibliography{literatur}


\end{document}